\pdfoutput=1
\documentclass[conference]{IEEEtran}
\IEEEoverridecommandlockouts

\usepackage{cite}
\usepackage{amsmath,amssymb,amsfonts}
\usepackage{algorithm}
\usepackage{algorithmic}
\usepackage{graphicx}
\usepackage{textcomp}
\usepackage[table]{xcolor}

\usepackage{booktabs}
\usepackage{array}
\usepackage{url}
\usepackage{placeins}
\definecolor{tablegray}{gray}{0.90}
\renewcommand{\algorithmicrequire}{\textbf{Input:}}
\renewcommand{\algorithmicensure}{\textbf{Output:}}

\newcounter{lemma}
\newcounter{proposition}
\newcounter{theorem}
\newcommand{\maybeparen}[1]{\if\relax\detokenize{#1}\relax\else~(#1)\fi}
\newenvironment{lemma}[1][]{\refstepcounter{lemma}\par\noindent\textbf{Lemma~\thelemma\maybeparen{#1}.}\ }{\par}
\newenvironment{proposition}[1][]{\refstepcounter{proposition}\par\noindent\textbf{Proposition~\theproposition\maybeparen{#1}.}\ }{\par}
\newenvironment{theorem}[1][]{\refstepcounter{theorem}\par\noindent\textbf{Theorem~\thetheorem\maybeparen{#1}.}\ }{\par}
\newenvironment{proof}{\par\noindent\textit{Proof.}}{\hfill$\square$\par}
\def\BibTeX{{\rm B\kern-.05em{\sc i\kern-.025em b}\kern-.08em
    T\kern-.1667em\lower.7ex\hbox{E}\kern-.125emX}}
\begin{document}

\title{AERIS: Offline Policy Improvement for Multi-UAV\\
Integrated Sensing and Communication}

\author{%
\IEEEauthorblockN{%
Ziyuan Wang\textsuperscript{1,\#},
Yifan Sui\textsuperscript{2,\#},
Wei Wei\textsuperscript{3},
Wenjie Xin\textsuperscript{4},\\
Zekai Zhang\textsuperscript{1},
Xiangwang Hou\textsuperscript{1},
Xiao-Ping (Steven) Zhang\textsuperscript{1,5}}
\IEEEauthorblockA{\footnotesize
\textsuperscript{1}Tsinghua University, Beijing, China
\hspace{1.5em}\textsuperscript{2}Shanghai Jiao Tong University, China\\
\textsuperscript{3}The University of Hong Kong, Hong Kong
\hspace{1.5em}\textsuperscript{4}The Hong Kong University of Science and Technology (Guangzhou), China\\
\textsuperscript{5}Toronto Metropolitan University, Canada\\
\textsuperscript{\#}Equal contribution.}}

\maketitle

\begin{abstract}
Unmanned aerial vehicle (UAV)-enabled integrated sensing and communication (ISAC) is a promising 6G paradigm, but dynamic multi-UAV ISAC control must jointly balance communication quality, sensing reliability, and flight safety under stochastic mobility. Existing optimization methods often require repeated global non-convex solving, while online reinforcement learning (RL) depends on risky trial-and-error flights that may cause sensing loss or collision-risk events.

This paper proposes AERIS, an offline policy improvement framework for multi-UAV ISAC. AERIS learns from fixed flight logs under centralized training and decentralized execution, so each UAV acts from local histories while training uses logged global information to assess team-level effects. We further design STAR-CRDT, an offline multi-agent RL algorithm that performs support-aware local action rectification and distills only trusted improvements into the decentralized actor. We prove an offline-support policy improvement guarantee. Experiments show that STAR-CRDT improves the main ISAC objective return by 29.3\% over the strongest baseline. It further improves communication sum rate, sensing pass rate, and sensing margin by 3.4\%, 4.8\%, and 69.1\%, while reducing collision-risk events by 54.2\%. On unseen real-road maps built from OpenStreetMap data, STAR-CRDT still obtains the best return.
\end{abstract}


\section{Introduction}

Integrated sensing and communication (ISAC) has become a key direction for sixth-generation (6G) networks because the same wireless infrastructure can deliver data services and sense the physical environment \cite{liu2022isac,zhang2021overview,liu2020joint_radar_comm}. This joint use of spectrum, antennas, and waveforms is especially valuable for applications that need both connectivity and environmental awareness, such as mobile target monitoring, emergency response, and temporary hotspot service. Unmanned aerial vehicles (UAVs) further enlarge this design space. By adjusting their positions and altitudes, UAVs can rapidly create line-of-sight links, approach sensing targets, and serve areas where fixed infrastructure is unavailable or overloaded \cite{zeng2016uav,zeng2019cellular,mozaffari2019tutorial}. Recent UAV-enabled ISAC studies therefore jointly optimize trajectory, beamforming, user association, and sensing quality \cite{lyu2023uav_isac,meng2022uav_periodic,gao2024marl_uav_isac,cheng2025radcom}.

The same flexibility also makes multi-UAV ISAC control difficult. A local movement of one UAV changes air-to-ground channels, beam directions, inter-user interference, sensing margins, and the distance to other UAVs. Ground users and sensing targets move stochastically, which turns the problem into a long-horizon control task. Classical trajectory and resource optimization methods usually rely on instantaneous global geometry, channel, and association information. Under stochastic mobility, they may also require repeated non-convex optimization or receding-horizon updates, which creates heavy computation and coordination overhead for distributed and fast UAV response \cite{zeng2019energy,mozaffari2017iot,wu2018uav_wpcn}. These properties make reinforcement learning (RL) and multi-agent RL (MARL) a natural tool for dynamic wireless control, because they can learn long-term policies from data without solving a full non-convex program at every slot \cite{mao2016resource,xu2018experience,luong2019deep_rl_wireless,lowe2017maddpg,foerster2018coma,rashid2018qmix}. Offline RL further extends this data-driven formulation to fixed logs, so experience collected from human-operated flights, legacy controllers, or conservative behavior policies can be reused for policy improvement \cite{fu2020d4rl,gulcehre2020rl_unplugged}. Representative offline methods have been developed for behavior-constrained value learning, conservative actor updates, and sequence-model-based control \cite{levine2020offline,fu2020d4rl,gulcehre2020rl_unplugged,fujimoto2019bcq,kumar2019bear,kumar2020cql,kostrikov2022iql,fujimoto2021td3bc,wang2020crr,chen2021decision_transformer}.

Although RL-based control for UAV networking and multi-UAV ISAC is gaining momentum, we observe that current RL solutions still have fundamental limitations that can make multi-UAV sensing and communication inefficient or unreliable. In the online setting, policy improvement still depends on trial rollouts. For multi-UAV ISAC, such rollouts are not just sampling costs. An immature or exploratory policy can lower sensing beam gain, degrade user service, or violate the inter-UAV safety distance before it becomes useful. Offline RL is attractive because it keeps the long-horizon policy-learning ability of RL while moving improvement to fixed data. Yet existing offline designs do not directly resolve the information and support structure of multi-UAV ISAC. A centralized learner can evaluate global state and joint actions, but a deployed UAV can only observe local history. Purely local or imitation-based learners may stay close to behavior logs, but they cannot reliably judge how a local action changes global interference, sensing margin, service quality, and pairwise safety, so residual ISAC violations may persist. More aggressive actor-critic rectification may leave the logged joint-action support and suffer from critic overestimation, which can further degrade sensing and safety \cite{pan2022omar,wang2023omiga}.

These limitations call for a new offline MARL design for multi-UAV ISAC. Such a design should improve beyond imperfect logs, keep each local correction support-aware, and allow distributed UAVs to improve task policies offline before deployment. To address this need, we propose Aerial Experience-guided Reinforcement for ISAC Systems (AERIS), a novel offline distributed policy improvement framework for multi-UAV ISAC. AERIS first turns logged flight experience into a centralized training and decentralized execution (CTDE) offline MARL problem, which avoids additional risky exploration while preserving local UAV execution. It then uses centralized training information to evaluate whether a local action change improves the team-level communication, sensing, and safety objective. To improve beyond behavior imitation without leaving reliable support, we design Support-aware Trust-gated Actor Rectification for Critic-Regularized Decision Transformer (STAR-CRDT). STAR-CRDT keeps a shared local-history actor for execution, searches for candidate corrections near logged and actor actions, scores them with a centralized critic, and distills a correction only when its critic gain and behavior proximity are reliable. Together, these choices turn centralized critic feedback into a support-aware teacher for the local actor, allowing AERIS to improve beyond behavior logs without degenerating into pure imitation or unsupported critic maximization.

Our contributions are summarized as follows.

\begin{itemize}
\item We propose AERIS, a novel offline policy improvement framework for stochastic long-horizon multi-UAV ISAC control. AERIS enables distributed UAVs to refine communication, sensing, and safety decisions from existing flight experience, reducing the need for risky online trial flights before deployment.
\item We design STAR-CRDT, a new offline MARL algorithm for support-aware local policy correction. By combining centralized critic scoring with trust-gated distillation, STAR-CRDT improves the decentralized actor while controlling support shift in fixed-log training.
\item We provide a proof of offline-support policy improvement and evaluate AERIS against mainstream offline RL/MARL baselines. The results show strong offline policy improvement, consistent system-scale gains, and the best return on unseen real-road maps built from OpenStreetMap (OSM) data.
\end{itemize}

\section{Background and Motivation}

This section presents two diagnostics that motivate offline policy improvement for multi-UAV ISAC. We first show why direct online policy improvement is risky, then examine whether existing fixed-log methods can remove this risk, and finally summarize the design requirements that guide AERIS.

\subsection{Online Trial Risk}

Multi-UAV ISAC requires a controller to trade off communication rate, sensing feasibility, sensing robustness, and flight safety. Existing studies model these objectives through coupled trajectory and beamforming variables \cite{lyu2023uav_isac,meng2022uav_periodic,gao2024marl_uav_isac,cheng2025radcom}. Online RL can adapt to the resulting stochastic dynamics, but it must collect trial rollouts. Safe and constrained RL studies have long recognized that unsafe exploration should be measured as a constraint cost rather than hidden inside average reward \cite{garcia2015safe,achiam2017cpo,chow2018lyapunov}.

Fig.~\ref{fig:motivation}(a) reports an online-learning diagnostic under an evaluation model aligned with recent multi-UAV ISAC trajectory and beamforming studies \cite{lyu2023uav_isac,meng2022uav_periodic,gao2024marl_uav_isac,cheng2025radcom}. Intermediate Twin Delayed Deep Deterministic Policy Gradient (TD3) checkpoints trigger collision-risk events before the controller becomes useful \cite{fujimoto2018td3}. After a behavior controller is obtained, additional Gaussian action noise keeps the communication sum rate roughly stable but increases sensing-threshold violations. Thus, a rollout can look acceptable from the communication metric while still damaging sensing reliability and safety. This makes continued online improvement undesirable for physical multi-UAV ISAC deployment.

\begin{figure}[t]
\centering
\begin{minipage}[t]{0.52\linewidth}
\centering
\includegraphics[width=\linewidth]{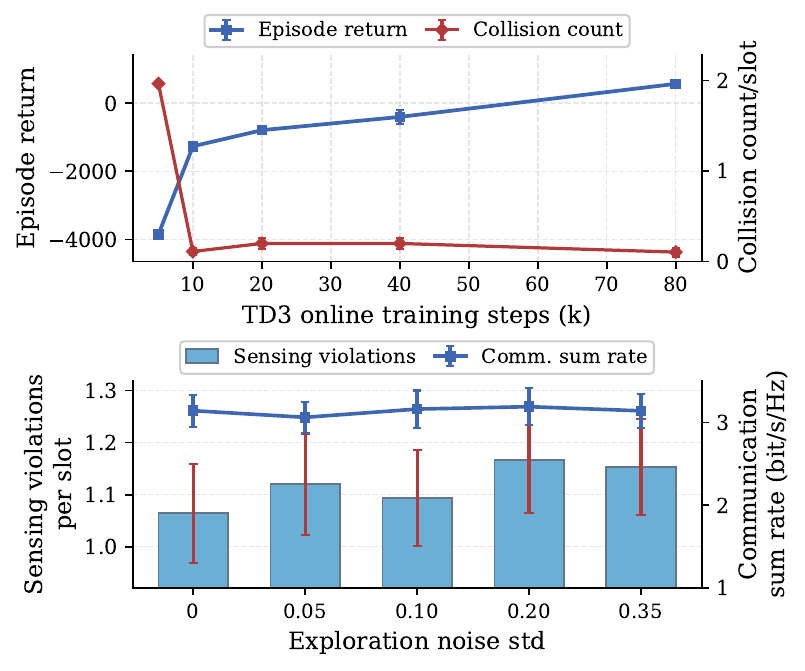}
\vspace{-1mm}
{\footnotesize (a) Risk during online improvement.}
\end{minipage}
\hfill
\begin{minipage}[t]{0.43\linewidth}
\centering
\includegraphics[width=\linewidth]{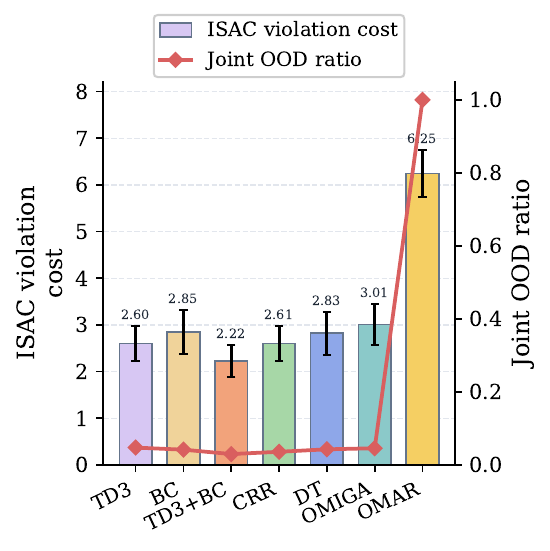}
\vspace{-1mm}
{\footnotesize (b) Support shift of offline baselines.}
\end{minipage}
\caption{Motivation diagnostics. Online trial actions can create sensing and safety risk. Generic offline baselines face an imitation and out-of-distribution (OOD) dilemma under fixed joint-action support.}
\label{fig:motivation}
\end{figure}

\subsection{Fixed-Log Baseline Diagnostic}

Offline RL is a natural candidate for avoiding new exploratory flights because it improves policies from a fixed dataset. Representative offline RL solutions have been widely adopted as fixed-log baselines in control and networking studies \cite{levine2020offline,fu2020d4rl,gulcehre2020rl_unplugged,fujimoto2021td3bc,wang2020crr,chen2021decision_transformer}. However, multi-UAV ISAC requires an offline update to improve a locally executed action while its effect is judged by the coupled team-level communication, sensing, and safety objective. This creates an information mismatch between centralized evaluation and decentralized execution. Conservative or imitation-based methods may inherit residual ISAC violations from the behavior policy, whereas aggressive actor rectification may leave the logged joint-action support and amplify critic error \cite{pan2022omar,wang2023omiga}.

Fig.~\ref{fig:motivation}(b) illustrates this tradeoff with representative fixed-log baselines. The bars report ISAC violation cost, which combines sensing-threshold violations and collision-risk events under the same penalties used by the reward. The line reports the joint-action OOD ratio measured by a $k$-nearest-neighbor support test against logged joint actions. Conservative methods remain closer to the log but still leave non-negligible ISAC violations, while aggressive rectification moves farther outside support and incurs larger violation cost. Therefore, directly porting state-of-the-art offline RL/MARL methods does not simultaneously provide policy improvement and reliable support control for multi-UAV ISAC.

\subsection{Motivation and Design Requirements}

The two diagnostics lead to three requirements for offline multi-UAV ISAC control.

\begin{itemize}
\item \textbf{Avoid repeated exploratory flights.} Policy improvement should avoid repeated exploratory flights after an initial log is available. If no public flight log exists, limited online learning can be used only to construct a behavior log, after which the dataset should be frozen and all policy refinement should happen offline.
\item \textbf{Preserve distributed UAV execution.} The trained controller must remain executable by distributed UAVs. Each UAV should act from its local history, while training may use logged global states and joint actions to evaluate team-level communication, sensing, and safety effects.
\item \textbf{Improve within reliable offline support.} The offline update should improve beyond imperfect behavior logs without trusting unsupported actions. This requires local corrections that are selected by a centralized critic but filtered by their proximity to the offline data support.
\end{itemize}

Taken together, these requirements motivate AERIS and highlight the main difficulty: offline policy improvement should use centralized training signals to identify team-level gains, yet convert them into a decentralized UAV controller without degenerating into pure imitation or unsupported critic maximization.

\section{System and Problem Formulation}

We adopt a standard UAV-enabled ISAC trajectory-and-beamforming model and formulate a generic communication-sensing optimization problem \cite{lyu2023uav_isac,meng2022uav_periodic}.

\subsection{System Model}

The system model describes the physical control environment in which moving UAVs serve communication users and sensing targets. Each movement or beam choice affects communication quality, sensing reliability, and safety separation over time.

We consider $M$ UAVs serving $K$ communication users and $L$ sensing targets in a $L_x\times L_y$ area. UAV $m$ has position $\mathbf q_m(t)=[x_m(t),y_m(t),z_m(t)]^T$. During log collection and random-mobility evaluation, users and targets follow bounded Brownian-style random walks over the full region \cite{camp2002survey}. This map-agnostic prior represents high-randomness mobility rather than a fixed road topology. During road-map deployment, entities still move stochastically, but every horizontal update is restricted to OSM road segments and Geographic Information System (GIS) layers. Dense grids induce frequent route branching, arterial corridors concentrate demand along a few directions, and sparse local roads create uneven service opportunities \cite{openstreetmap_copyright,geofabrik_osm_gis}.

\begin{figure}[t]
\centering
\includegraphics[width=0.78\linewidth]{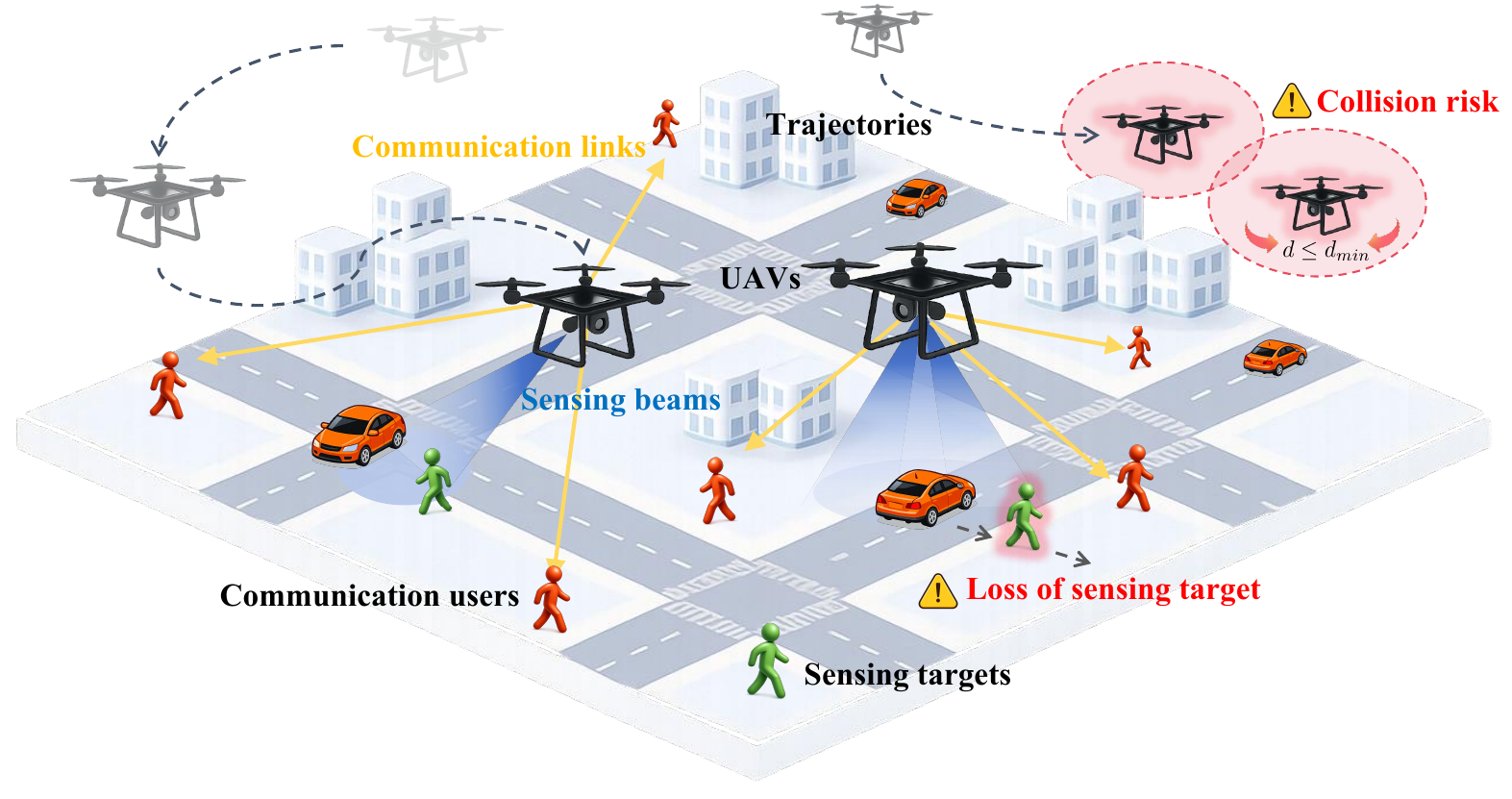}
\caption{Dynamic multi-UAV ISAC control under local execution.}
\label{fig:system}
\end{figure}

Each UAV uses an $N_t$-element transmit array for downlink communication and sensing. The air-to-ground channel $\mathbf h_{m,k}(t)$ captures path loss, low-altitude attenuation, and array steering, as widely used for UAV links \cite{alhourani2014optimal,zeng2019energy}. With beamformer $\mathbf w_{m,k}(t)$ steered from UAV $m$ to user $k$, the received signal-to-interference-plus-noise ratio (SINR) is
\begin{equation}
\gamma_{m,k}(t)=
\frac{|\mathbf h_{m,k}^H(t)\mathbf w_{m,k}(t)|^2}
{\sigma^2+\sum_{(i,j)\ne(m,k)}|\mathbf h_{i,k}^H(t)\mathbf w_{i,j}(t)|^2},
\label{eq:comm_utility}
\end{equation}
where $\sigma^2$ is the noise power. The communication rate is $R_{m,k}(t)=B\log_2(1+\gamma_{m,k}(t))$, and $U_{\rm comm}(t)=\sum_m\sum_{k\in\mathcal C_m(t)}R_{m,k}(t)$ is the system sum-rate. This standard metric directly improves spectral efficiency in UAV-enabled ISAC trajectory and beamforming designs \cite{lyu2023uav_isac,meng2022uav_periodic}. For sensing, let $\mathbf R_m(t)=\sum_{k\in\mathcal C_m(t)}\mathbf w_{m,k}(t)\mathbf w_{m,k}^H(t)$ be the transmit covariance. The beampattern gain $G_{m,\ell}(t)=\mathbf a_{m,\ell}^H(t)\mathbf R_m(t)\mathbf a_{m,\ell}(t)$ measures energy focused toward target $\ell$, and $\Delta_{m,\ell}^{\rm sen}(t)=G_{m,\ell}(t)-\Gamma_{m,\ell}(t)$ is its threshold margin. The sensing utility is
\begin{align}
U_{\rm sen}(t)=&\;\alpha_p\sum_{m,\ell}\mathbb I\{\Delta_{m,\ell}^{\rm sen}(t)\ge0\} \nonumber\\
&+\alpha_m\sum_{m,\ell}
\tanh\!\left(\frac{\Delta_{m,\ell}^{\rm sen}(t)}
{\max(\Gamma_{m,\ell}(t),\varepsilon)}\right),
\label{eq:sen_utility}
\end{align}
where the two terms reward target-threshold satisfaction and normalized beampattern margin. Optimizing this form keeps targets detectable and preserves robustness against mobility-induced geometry changes \cite{lyu2023uav_isac,meng2022uav_periodic,gao2024marl_uav_isac}.

\subsection{Optimization Problem}

The optimization asks each UAV to decide where to fly and how to form beams over time, so that users receive high-rate service, sensing targets remain detectable, and UAVs stay safely separated.

Under distributed execution, the local information available to UAV $m$ includes its own three-dimensional (3-D) position, relative geometry to associated communication users and sensing targets, and local channel and sensing steering estimates, without instantaneous global state exchange. The decision variables are the UAV trajectories $\mathbf q$ and beamformers $\mathbf w$ over horizon $T$. The objective is to maximize a weighted communication and sensing utility that evaluates the locally executed decisions by communication rate, sensing feasibility, and sensing robustness \cite{liu2022isac,lyu2023uav_isac,meng2022uav_periodic},
\begin{align}
\textbf{P0:}\quad
\max_{\{\mathbf q,\mathbf w\}}\;&
\sum_{t=0}^{T-1}\bigl[\lambda_c U_{\rm comm}(t)+\lambda_s U_{\rm sen}(t)\bigr]
\label{eq:p0}\\
\text{s.t.}\quad
&G_{m,\ell}(t)\ge \Gamma_{m,\ell}(t),\quad \forall m,\ell,t, \nonumber\\
&\sum_{k\in\mathcal C_m(t)}\|\mathbf w_{m,k}(t)\|_2^2\le P_{\max},
\quad \forall m,t, \nonumber\\
&\|\mathbf q_m(t)-\mathbf q_n(t)\|_2\ge d_{\min},\quad
\forall m\ne n,t, \nonumber
\end{align}
where the constraints enforce sensing-threshold satisfaction, per-UAV transmit power, and inter-UAV safety separation. Sensing-threshold violations can make returned data too weak to maintain target awareness \cite{lyu2023uav_isac,meng2022uav_periodic}. A collision-risk event occurs when pairwise distance falls below $d_{\min}$, and standard UAV kinematic limits are enforced during trajectory generation \cite{zeng2019energy,mozaffari2017iot}.

\subsection{Key Challenges}

Solving \textbf{P0} presents three key challenges.

\begin{itemize}
\item \textbf{Challenge 1: Stochastic non-convex ISAC control.} The SINR and beampattern terms jointly depend on UAV positions, beamformers, interference, and geometry-dependent associations. User and target mobility further makes \textbf{P0} a long-horizon stochastic control problem, so repeatedly solving a centralized trajectory-and-beamforming program is impractical for fast distributed UAV response.
\item \textbf{Challenge 2: Online policy-improvement risk.} A learning controller for \textbf{P0} would normally require trial rollouts to estimate long-term rewards and constraint effects. In physical multi-UAV ISAC, such exploratory flights can consume flight resources, reduce communication service, weaken sensing returns, or create collision-risk events before the policy converges.
\item \textbf{Challenge 3: Fixed-log distributed improvement.} Freezing a flight log removes new exploration, but it also makes policy improvement support-sensitive. Training can use logged global states and joint actions to evaluate team-level communication, sensing, and safety effects, whereas each deployed UAV must act from local histories. Thus, an offline learner must convert centralized training evidence into decentralized policy updates without imitating residual violations or trusting unsupported joint actions.
\end{itemize}

\section{AERIS Design}

AERIS addresses these challenges through a fixed-log CTDE design. It reformulates the stochastic non-convex control problem as sequential local action selection, preserving distributed execution while allowing training to evaluate team-level ISAC rewards. It then freezes flight experience collected by a conservative behavior controller, so policy refinement does not require additional exploratory flights. Finally, it uses STAR-CRDT to turn centralized critic feedback into support-aware updates of the local-history actor. Thus, AERIS defines the policy-improvement pipeline, while STAR-CRDT is the offline MARL optimizer inside it.

\subsection{Problem Transformation and Offline Dataset}

AERIS transforms \textbf{P0} from direct trajectory-and-beamforming optimization into distributed agent action selection. At slot $t$, the local observation $o_t^m$ of UAV $m$ corresponds to the local information described in Section~III-B. A short local history $h_t^m$ stacks recent observations, previous local actions, and return-conditioning signals. The decentralized actor outputs $a_t^m=\pi_\theta(h_t^m)$, where $a_t^m$ controls UAV displacement and beam residuals, and $\mathbf a_t=(a_t^1,\ldots,a_t^M)$ is the joint action. This conversion preserves distributed execution because each UAV acts only from local information, while the training process may still use global logged states to evaluate the joint effect of the selected actions.

Under this RL reformulation, AERIS optimizes
\begin{equation}
\textbf{P1:}\quad
\max_{\pi}\mathbb E_{\pi}\Bigl[\sum_{t=0}^{T-1}r_t\Bigr],
\label{eq:p1}
\end{equation}
with per-slot reward
\begin{align}
r_t={}&\lambda_r U_{\rm comm}(t)+\lambda_p P_{\rm pass}(t)
+\lambda_m M_{\rm sen}(t) \nonumber\\
&-\zeta N_t^{\rm vio}-c_{\rm col}N_t^{\rm col},
\label{eq:reward}
\end{align}
where the positive terms are slot-level slices of \textbf{P0}. $U_{\rm comm}(t)$ rewards communication spectral efficiency, $P_{\rm pass}(t)$ counts sensing-threshold passes, and $M_{\rm sen}(t)$ measures normalized sensing margin. The negative terms relax the constraints in \textbf{P0}, where $N_t^{\rm vio}$ counts sensing-threshold violations and $N_t^{\rm col}$ counts collision-risk events. Such reward-and-penalty reformulation is standard in RL-based networking and constrained RL \cite{mao2016resource,xu2018experience,luong2019deep_rl_wireless,achiam2017cpo}. It keeps the communication, sensing, and safety meanings of \textbf{P0}, while making the problem suitable for fixed-log policy learning.

There is no public flight-log dataset for multi-UAV ISAC with joint communication, sensing, and safety labels. We therefore collect a fixed behavior dataset before offline training, following common offline RL practice \cite{fu2020d4rl,gulcehre2020rl_unplugged}. The behavior policy uses TD3 under the same action and reward semantics as \textbf{P1}. TD3 is suitable here because it is a strong continuous-control actor-critic method, and TD3 has been widely used as a continuous-control backbone for communication and networking optimization problems \cite{fujimoto2018td3,luong2019deep_rl_wireless}. With clipped target noise $\epsilon$, TD3 uses
\begin{equation}
y_t^{\rm TD3}=r_t+\gamma\min_{j=1,2}Q_{\bar\phi_j}
\bigl(x_{t+1},\pi_{\bar\theta}(x_{t+1})+\epsilon\bigr),
\label{eq:td3_target}
\end{equation}
to fit twin critics and update the actor by delayed deterministic policy-gradient steps \cite{fujimoto2018td3}. During collection, each UAV uses only local geometry, service information, neighbor separation, and recent rewards, so the resulting log respects the distributed information pattern.

After the preliminary controller is obtained, AERIS treats its rollouts as surrogate flight logs. They play the role of human-pilot or legacy-controller records, which contain useful behavior but may retain sensing and safety violations. The dataset is then frozen. Each trajectory $\tau\in\mathcal D$ stores $\tau=\{(s_t,\mathbf a_t,r_t,s_{t+1},\{h_t^m\}_{m=1}^M)\}_{t=0}^{T-1}$, where $s_t$ and $\mathbf a_t$ support centralized critic training and $h_t^m$ supports decentralized actor learning. TD3 only constructs the log, whereas all reported AERIS policy improvement comes from offline learning.

\subsection{STAR-CRDT Offline Learning}

STAR-CRDT follows one principle. A UAV should change a logged action only when the centralized critic predicts a global ISAC improvement and the change remains close to reliable offline support. Fig.~\ref{fig:workflow} illustrates this procedure. The actor first learns a local sequence policy from the log, the centralized critic evaluates local corrections in joint-action contexts, and the trust gate decides how much correction can be distilled back into the decentralized actor.

\begin{figure}[t]
\centering
\includegraphics[width=\linewidth]{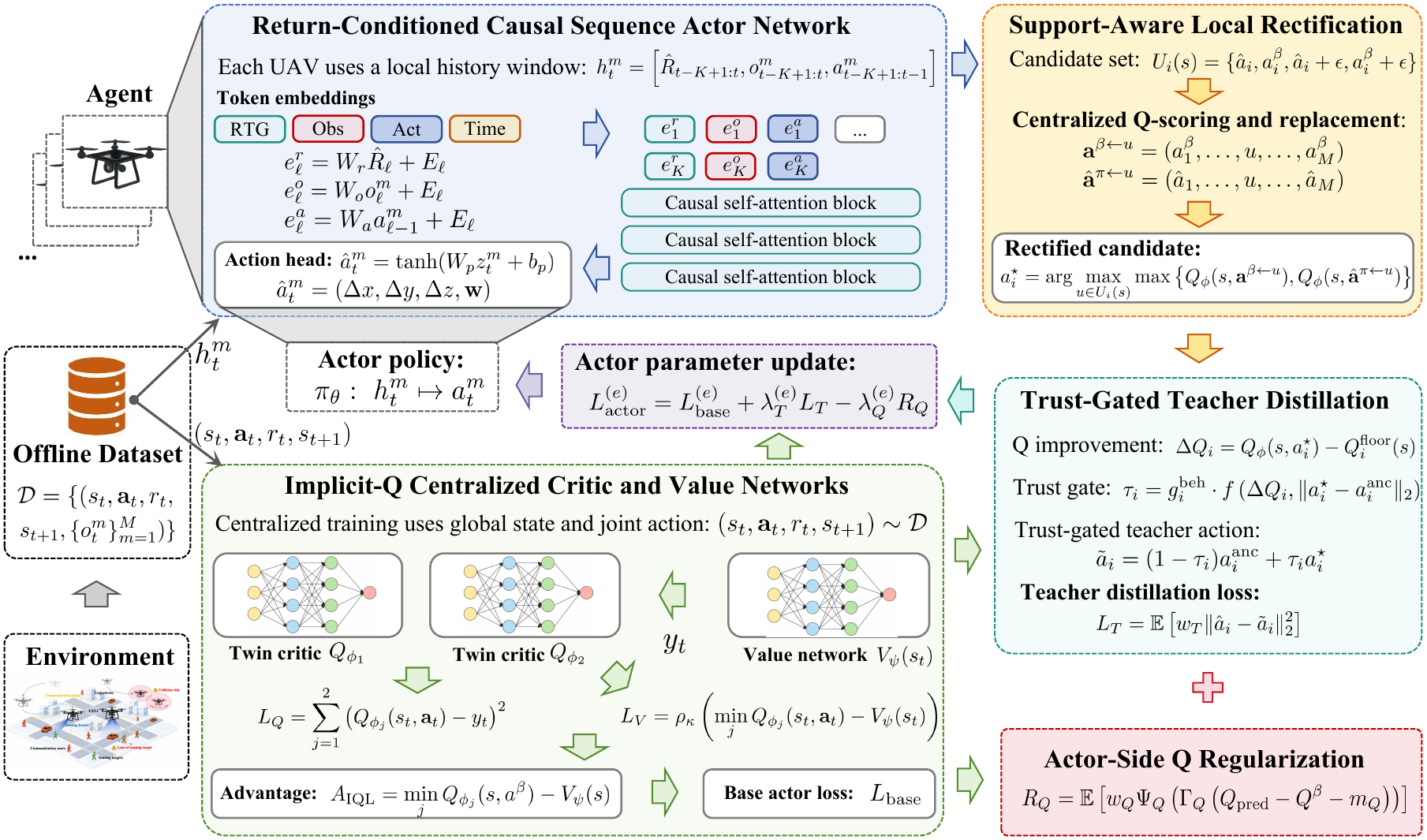}
\caption{Workflow of STAR-CRDT. The actor executes locally, while the centralized critic selects and filters local corrections during offline training.}
\label{fig:workflow}
\end{figure}

\begin{algorithm}[t]
\caption{STAR-CRDT Offline Training}
\label{alg:star_crdt}
\begin{algorithmic}[1]
\STATE \algorithmicrequire{} offline dataset $\mathcal D$, shared actor $\pi_\theta$, twin critics $Q_{\phi_1},Q_{\phi_2}$, value network $V_\psi$
\STATE \algorithmicensure{} decentralized STAR-CRDT actor $\pi_\theta$ and the best validation checkpoint
\FOR{each training epoch $e$}
\STATE sample mini-batches of fixed-length local-history windows from $\mathcal D$
\STATE run the shared actor on each local history to obtain $\widehat a_i=\pi_\theta(h_i)$
\STATE update $Q_{\phi_1},Q_{\phi_2}$ by minimizing \eqref{eq:q_loss}
\STATE update $V_\psi$ using the expectile value loss
\STATE compute the backbone actor loss $\mathcal L_{\rm base}^{(e)}$
\FOR{each selected token and UAV $i$}
\STATE construct $\mathcal U_i(s)$ by \eqref{eq:candidate_set}
\STATE evaluate behavior and actor contexts by the centralized critic
\STATE select the rectified action $a_i^\star$ by \eqref{eq:rect_action}
\STATE choose $a_i^{\rm anc}$ and construct $\widetilde a_i$ by \eqref{eq:trust}
\ENDFOR
\STATE build teacher targets and compute $\mathcal L_T^{(e)}$
\STATE compute the centralized Q-regularizer $\mathcal R_Q^{(e)}$
\STATE update the actor with \eqref{eq:actor_obj}
\STATE update target networks and record the best checkpoint
\ENDFOR
\end{algorithmic}
\end{algorithm}

\textbf{Local sequence actor and centralized critic.} Each UAV uses the same return-conditioned causal Transformer actor $\pi_\theta(h_t^m)$, inspired by Decision Transformer \cite{chen2021decision_transformer}. Parameter sharing lets all UAV logs train a common local rule across different local geometries. The actor is deliberately local, since a globally conditioned actor would require real-time global information exchange. However, local execution alone cannot judge whether a local correction improves the team objective. STAR-CRDT therefore trains centralized twin critics $Q_{\phi_1}(s_t,\mathbf a_t)$ and $Q_{\phi_2}(s_t,\mathbf a_t)$, together with a value network $V_\psi(s_t)$, using logged global states and joint actions. This follows the CTDE principle \cite{lowe2017maddpg,foerster2018coma,rashid2018qmix}. The critic evaluates global communication, sensing, and safety effects during training, while deployment uses only the local actor in Fig.~\ref{fig:workflow}.

\textbf{Centralized implicit-Q evaluation.} The next issue is conservative value estimation under fixed data. Directly maximizing a learned offline critic can select unsupported actions with overestimated values \cite{levine2020offline,kumar2020cql}. STAR-CRDT instead learns a ranking signal on logged joint actions. For each transition, STAR-CRDT minimizes
\begin{equation}
\mathcal L_Q=\frac12\mathbb E_{\mathcal D}\sum_{j=1}^2\bigl(Q_{\phi_j}(s_t,\mathbf a_t)-y_t\bigr)^2,
\label{eq:q_loss}
\end{equation}
where $y_t=r_t+\gamma V_\psi(s_{t+1})$. With $\bar Q_t=\min_j Q_{\phi_j}(s_t,\mathbf a_t)$, the value loss is $\mathcal L_V=\mathbb E_{\mathcal D}[\rho_\kappa(\bar Q_t-V_\psi(s_t))]$, where $\rho_\kappa(u)=|\kappa-\mathbb I\{u<0\}|u^2$ follows implicit Q-learning \cite{kostrikov2022iql}. The twin-critic minimum suppresses optimistic error, and the expectile value provides an in-support baseline.

\textbf{Behavior-preserving backbone.} Before critic-guided correction, the actor first learns a feasible local behavior. For a logged action $\mathbf a^\beta$, STAR-CRDT computes
\begin{equation}
A_{\rm IQL}(s,\mathbf a^\beta)=\min_j Q_{\phi_j}(s,\mathbf a^\beta)-V_\psi(s),
\label{eq:iql_adv}
\end{equation}
and $w_{\rm IQL}=\min\{\exp(\beta_{\rm IQL}[A_{\rm IQL}]_+),w_{\max}^{\rm IQL}\}$. The backbone loss mixes plain behavior cloning and advantage-weighted behavior cloning. This anchors the sequence actor near logged support before any correction is introduced.

\textbf{Support-aware local rectification.} A behavior-preserving actor can still inherit residual ISAC violations. STAR-CRDT therefore searches for local corrections only near the actor and behavior actions as follows.
\begin{align}
\mathcal U_i(s)=&\{\hat a_i,a_i^\beta\}\cup
\{\mathrm{clip}(\hat a_i+\sigma_\pi\xi_k)\}_{k=1}^{N_r} \nonumber\\
&\cup\{\mathrm{clip}(a_i^\beta+\sigma_\beta\zeta_k)\}_{k=1}^{N_r},
\label{eq:candidate_set}
\end{align}
where $\xi_k$ and $\zeta_k$ are local perturbations around the actor and behavior actions. Each candidate $u$ is tested in two centralized contexts. Let $\mathbf a_u^\beta$ replace only the $i$-th behavior action and $\mathbf a_u^\pi$ replace only the $i$-th actor action. With $Q_\phi=\min_jQ_{\phi_j}$,
\begin{equation}
a_i^\star\in\arg\max_{u\in\mathcal U_i(s)}\max_{\chi\in\{\beta,\pi\}}Q_\phi(s,\mathbf a_u^\chi).
\label{eq:rect_action}
\end{equation}
Thus, a local correction is accepted only after the centralized critic evaluates its effect on the joint ISAC return.

\textbf{Trust-gated distillation.} A critic-improving candidate can still be unreliable if its gain is small or its distance from support is large. STAR-CRDT therefore distills a conservative teacher. Let $a_i^{\rm anc}$ be the anchor action and $\Delta Q_i=Q_\phi(s,\mathbf a^\star)-Q_i^{\rm floor}(s)$. The trust coefficient is
\begin{equation}
\tau_i=\tau_{\max}\bigl(1-e^{-\beta_\tau[\Delta Q_i-m_\tau]_+}\bigr)e^{-\|a_i^\star-a_i^{\rm anc}\|_2^2/\sigma_\tau}.
\label{eq:trust}
\end{equation}
The teacher is $\widetilde a_i=(1-\tau_i)a_i^{\rm anc}+\tau_i a_i^\star$. The actor is updated by
\begin{equation}
\mathcal L_{\rm actor}^{(e)}=\mathcal L_{\rm base}^{(e)}+\lambda_T^{(e)}\mathcal L_T-\lambda_Q^{(e)}\mathcal R_Q,
\label{eq:actor_obj}
\end{equation}
where $\mathcal L_T$ distills the teacher and $\mathcal R_Q$ rewards predicted joint actions only when they improve over logged behavior. Algorithm~\ref{alg:star_crdt} summarizes the workflow. During training, STAR-CRDT repeatedly samples fixed-log local histories, updates centralized critics and value estimates, constructs support-aware local teachers, and distills them into the shared actor. After training, only $\pi_\theta(h_t^m)$ is deployed on UAVs. All critic, value, candidate-search, and trust-gate modules are removed from runtime control, so each UAV can make autonomous real-time decisions from its own local history.

\subsection{Policy Improvement Proof}

We prove the trust-gated local correction on the offline-support distribution, rather than claiming global optimality outside the dataset. Let $d_\beta$ be the state distribution induced by the behavior log and let $\mathbf a_{-i}^{\rm ctx}$ be the fixed context actions when only UAV $i$ is varied. The local rectification gain is
\begin{equation}
\Delta_i(s)=Q_\phi(s,(a_i^\star,\mathbf a_{-i}^{\rm ctx}))-Q_\phi(s,(a_i^{\rm anc},\mathbf a_{-i}^{\rm ctx})).
\label{eq:gain}
\end{equation}

Following standard offline actor-critic and smooth optimization analyses \cite{levine2020offline,kumar2020cql,kostrikov2022iql,boyd2004convex}, we use four local assumptions.

\textbf{Assumption 1.} The critic correctly ranks candidates in $\mathcal U_i(s)$, and the anchor satisfies $a_i^{\rm anc}\in\mathcal U_i(s)$.

\textbf{Assumption 2.} The local critic slice
\begin{equation}
g_i(u;s)=Q_\phi(s,(u,\mathbf a_{-i}^{\rm ctx})),
\label{eq:local_q_map}
\end{equation}
is concave on the segment between $a_i^{\rm anc}$ and $a_i^\star$.

\textbf{Assumption 3.} The post-update actor $a_i^+=\pi_{\theta^+}(h_i)$ satisfies $\mathbb E_{s\sim d_\beta}\|a_i^+-\widetilde a_i\|_2\le\varepsilon$.

\textbf{Assumption 4.} For all relevant $u$ and $v$,
\begin{equation}
|g_i(u;s)-g_i(v;s)|\le L_Q\|u-v\|_2.
\label{eq:lipschitz_assumption}
\end{equation}
Assumption 1 is the usual local ranking condition for using a learned offline critic only within supported candidates. Assumptions 2 and 4 are local regularity conditions on the short segment used by the trust gate, not global claims about the non-convex wireless objective. Assumption 3 states that supervised distillation tracks the teacher with bounded approximation error.

\begin{lemma}[Nonnegative rectification gain]
For every offline-support state $s$, $\Delta_i(s)\ge0$.
\end{lemma}
\begin{proof}
The selected action $a_i^\star$ maximizes the critic value over $\mathcal U_i(s)$, and $a_i^{\rm anc}\in\mathcal U_i(s)$ by Assumption 1. Hence the selected action cannot have a smaller critic value than the anchor, which gives \eqref{eq:gain}.
\end{proof}

Lemma 1 shows why local corrections are evaluated in a joint-action context. The candidate accepted by STAR-CRDT is never worse than the anchor under the centralized critic.

\begin{proposition}[Gain of the trust-gated teacher]
For $\widetilde a_i=(1-\tau_i)a_i^{\rm anc}+\tau_i a_i^\star$ and $0\le\tau_i\le1$,
\begin{equation}
Q_\phi(s,(\widetilde a_i,\mathbf a_{-i}^{\rm ctx}))\ge Q_\phi(s,(a_i^{\rm anc},\mathbf a_{-i}^{\rm ctx}))+\tau_i\Delta_i(s).
\label{eq:teacher_gain}
\end{equation}
\end{proposition}
\begin{proof}
By concavity on the segment between $a_i^{\rm anc}$ and $a_i^\star$, Jensen's inequality gives
\begin{align}
Q_\phi(s,(\widetilde a_i,\mathbf a_{-i}^{\rm ctx}))
&\ge(1-\tau_i)Q_\phi(s,(a_i^{\rm anc},\mathbf a_{-i}^{\rm ctx})) \nonumber\\
&\quad+\tau_iQ_\phi(s,(a_i^\star,\mathbf a_{-i}^{\rm ctx})),
\end{align}
which is equivalent to \eqref{eq:teacher_gain}.
\end{proof}

Proposition 1 explains the trust gate. A partial move toward the rectified action preserves a proportional critic gain, so the teacher can improve the anchor while remaining conservative.

\begin{theorem}[Offline-support policy improvement]
If $\mathbb E_{s\sim d_\beta}\|a_i^+-\widetilde a_i\|_2\le\varepsilon$, then
\begin{align}
\mathbb E_{d_\beta}Q_\phi(s,(a_i^+,\mathbf a_{-i}^{\rm ctx}))
\ge&\;\mathbb E_{d_\beta}Q_\phi(s,(a_i^{\rm anc},\mathbf a_{-i}^{\rm ctx})) \nonumber\\
&+\mathbb E_{d_\beta}[\tau_i\Delta_i(s)]-L_Q\varepsilon.
\label{eq:improvement}
\end{align}
Thus, when $\mathbb E_{d_\beta}[\tau_i\Delta_i(s)]>L_Q\varepsilon$, the actor update improves the centralized-critic surrogate in expectation.
\end{theorem}
\begin{proof}
By Assumption 4,
\begin{equation}
Q_\phi(s,(a_i^+,\mathbf a_{-i}^{\rm ctx}))\ge Q_\phi(s,(\widetilde a_i,\mathbf a_{-i}^{\rm ctx}))-L_Q\|a_i^+-\widetilde a_i\|_2.
\end{equation}
Taking expectation over $d_\beta$, applying Assumption 3, and using Proposition 1 proves \eqref{eq:improvement}.
\end{proof}

Candidate search first produces a nonnegative centralized-critic gain over the anchor action, the trust gate admits only the supported portion of this gain, and teacher distillation transfers the resulting correction to the local actor. Under the stated local ranking, regularity, and distillation-error conditions, if the expected trusted rectification gain exceeds $L_Q\varepsilon$, the post-update actor has a larger centralized-critic surrogate value than the anchor in expectation over $d_\beta$. This gives the offline-support policy-improvement guarantee used by AERIS.

\begin{figure*}[!t]
\centering
\includegraphics[width=\textwidth]{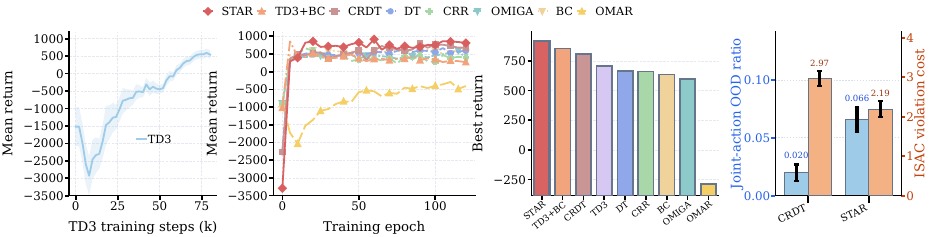}
\caption{Training and ablation results under random mobility. The first three panels show TD3 behavior-policy training, offline policy evaluation curves, and best-checkpoint return comparison. The fourth panel compares CRDT and STAR-CRDT in joint-action OOD ratio and ISAC violation cost. STAR-CRDT reduces violations while keeping critic-corrected actions near useful support.}
\label{fig:train_support}
\end{figure*}

\begin{figure*}[!t]
\centering
\includegraphics[width=\textwidth]{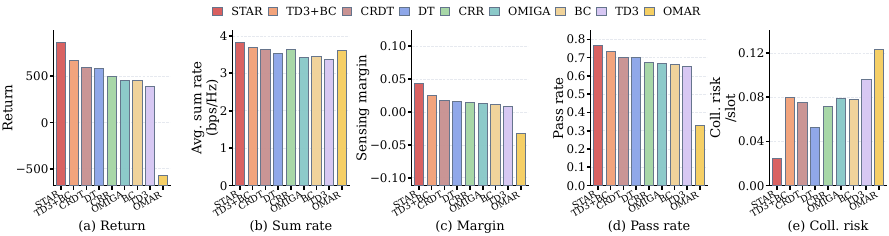}
\caption{Random-mobility offline evaluation. The compact five-panel row reports return, communication sum rate, sensing margin, sensing pass rate, and collision-risk count. STAR-CRDT obtains the best overall tradeoff.}
\label{fig:random_metrics}
\end{figure*}

\begin{figure*}[!t]
\centering
\includegraphics[width=\textwidth]{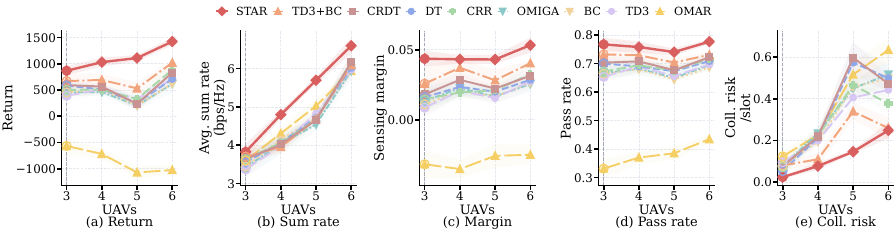}
\caption{System-scale evaluation under random mobility. The compact five-panel row follows Fig.~\ref{fig:random_metrics} and reports return, communication sum rate, sensing margin, sensing pass rate, and collision-risk count. STAR-CRDT remains strongest when the trained shared actor is evaluated with more UAVs.}
\label{fig:scale}
\end{figure*}

\section{Experimental Evaluation}

\subsection{Evaluation Setup}

We evaluate AERIS under random-mobility and road-map deployment protocols. The first isolates fixed-log offline improvement under the bounded Brownian model used for log collection, while the second freezes each learned policy and transfers it without road-specific tuning to OSM road-constrained mobility. Thus, random mobility supplies a high-randomness source log, while road-map deployment mimics the common mismatch between available offline logs and actual application scenarios.

The evaluation instantiates the model in Sections~III-A and III-B with the parameters in Table~\ref{tab:params}. By default, $M=3$ UAVs serve $K=6$ users and $L=3$ targets, with two users and one target per UAV. Service sets use quota-based nearest-UAV association, assigning each user or target to the closest UAV that has not reached its user or target cap \cite{lyu2023uav_isac,meng2022uav_periodic,gao2024marl_uav_isac}. The scale study varies $M$ from 3 to 6 UAVs to test coordination as the fleet size increases \cite{gao2024marl_uav_isac,cheng2025radcom}. For high-risk low-altitude urban operation, $d_{\min}=30$ m follows conservative multi-UAV ISAC/RadCom safety settings \cite{gao2024marl_uav_isac,cheng2025radcom}. Other values follow UAV-link, ISAC beamforming, mobility/safety, and OSM/GIS road-network studies \cite{alhourani2014optimal,zeng2019energy,lyu2023uav_isac,meng2022uav_periodic,camp2002survey,boeing2017osmnx,openstreetmap_copyright,geofabrik_osm_gis}.

The fixed log is collected once by the preliminary TD3 behavior controller and then held fixed. All offline methods receive the same trajectories, rewards, local histories, global states for centralized training, and decentralized observations, so the comparison measures improvement from the same support without extra online interaction or road-map adaptation.

\textbf{Baselines.} We compare with representative offline RL and offline MARL methods spanning imitation, conservative updates, sequence modeling, and multi-agent actor rectification \cite{levine2020offline,gulcehre2020rl_unplugged,luong2019deep_rl_wireless}. The dataset, training budget, model capacity where applicable, and evaluation episodes are kept consistent.
\begin{itemize}
\setlength{\itemsep}{0pt}
\setlength{\parsep}{0pt}
\setlength{\topsep}{1pt}
\item \textbf{Behavior cloning (BC)}~\cite{pomerleau1989alvinn} tests whether pure supervised imitation of logged actions is sufficient.
\item \textbf{DT} is a Decision Transformer baseline~\cite{chen2021decision_transformer} that learns return-conditioned decentralized actions from local histories.
\item \textbf{TD3+BC}~\cite{fujimoto2021td3bc} and \textbf{CRR} (Critic Regularized Regression)~\cite{wang2020crr} represent conservative offline actor updates.
\item \textbf{CRDT} is our sequence-model ablation. It keeps the local Decision-Transformer actor~\cite{chen2021decision_transformer} and centralized critic regularization, but removes support-aware candidate search and trust-gated distillation. Comparing CRDT with STAR-CRDT isolates whether trust-gated local rectification is necessary beyond critic regularization alone.
\item \textbf{OMIGA}~\cite{wang2023omiga} and \textbf{OMAR}~\cite{pan2022omar} are offline MARL baselines that respectively emphasize global-to-local value regularization and actor rectification.
\item \textbf{Online TD3}~\cite{fujimoto2018td3} is reported as the behavior-policy reference rather than an additional online competitor.
\end{itemize}
\vspace{-0.35em}

\textbf{Evaluation Metrics.} We report return $J=\sum_{t=0}^{T-1}r_t$, average communication sum rate $R_{\rm sum}=T^{-1}\sum_t U_{\rm comm}(t)$, sensing pass rate (fraction of sensing links satisfying thresholds), sensing margin (average normalized excess over thresholds), and collision-risk count $C_{\rm col}=T^{-1}\sum_t N_t^{\rm col}$, where lower is better. Effective improvement should jointly improve return, communication, sensing, and safety.

\begin{table}[t]
\vspace{-0.45em}
\centering
\caption{Key Parameters}
\label{tab:params}
\scriptsize
\renewcommand{\arraystretch}{1.03}
\begin{tabular*}{\linewidth}{@{\extracolsep{\fill}}p{0.50\linewidth}p{0.41\linewidth}@{}}
\specialrule{\heavyrulewidth}{0pt}{0pt}
\rowcolor{tablegray}\multicolumn{2}{c}{\textbf{System parameters}}\\
\specialrule{\lightrulewidth}{0pt}{0pt}
Service area $L_x\times L_y$ & $500\times500$ m$^2$ \\
Episode horizon $T$ & 200 slots \\
Main topology $(M,K,L)$ & $(3,6,3)$ \\
Scale sweep $M$ & 3 to 6 UAVs \\
UAV altitude $[z_{\min},z_{\max}]$ & $[100,150]$ m \\
UAV speed $V_{\max}$ & 5 m/slot \\
Ground speed $v_u^{\max}$ & 0.5 m/slot \\
Safety distance $d_{\min}$ & 30 m \\
Array size $N_t$ & 4 antennas \\
Power budget $P_{\max}$ & 20 dBm \\
Noise power $\sigma^2$ & $10^{-4}$ \\
Bandwidth $B$ & 1 normalized Hz \\
Sensing threshold $\Gamma$ & $2\times10^{-6}$ \\
Log size & 1000 episodes, 200k transitions \\
Held-out evaluation & 50 episodes \\
\specialrule{\lightrulewidth}{0pt}{0pt}
\rowcolor{tablegray}\multicolumn{2}{c}{\textbf{STAR-CRDT parameters}}\\
\specialrule{\lightrulewidth}{0pt}{0pt}
Actor context length & $K=50$ slots \\
Actor hidden size & 256 \\
Transformer blocks & 6 layers, 8 heads \\
Critic hidden size & 512 \\
Discount and expectile & $\gamma=0.99$, $\kappa=0.8$ \\
Q regularization $\lambda_Q$ & 0.18 \\
Rectification samples $N_r$ & 4 per anchor \\
Rectification std. & $\sigma_\pi=0.18$, $\sigma_\beta=0.06$ \\
Trust range & 0.35 to 0.45 \\
\bottomrule
\end{tabular*}
\vspace{-0.9em}
\end{table}

\begin{table*}[!t]
\centering
\providecommand{\best}[1]{\textcolor{red}{\textbf{\boldmath #1}}}
\providecommand{\second}[1]{\textcolor{blue}{\underline{\textbf{\boldmath #1}}}}
\providecommand{\alg}[1]{\textbf{\textit{#1}}}
\providecommand{\metric}[1]{\textbf{\textit{#1}}}
\caption{Zero-shot road-map deployment evaluation. Metrics are shown by row for each road scene. Best values are bold red, second-best values are underlined bold blue, and lower is better only for Coll. risk.}
\label{tab:road_eval}
\tiny
\setlength{\tabcolsep}{1.05pt}
\renewcommand{\arraystretch}{0.82}
\resizebox{\textwidth}{!}{%
\begin{tabular}{llccccccccc ccccccccc}
\toprule
\metric{Scene} & \metric{Metric} & \multicolumn{9}{c}{\textbf{$M=3$}} & \multicolumn{9}{c}{\textbf{$M=6$}} \\
\cmidrule(lr){3-11}\cmidrule(lr){12-20}
 & & \alg{STAR} & \alg{TD3+BC} & \alg{CRDT} & \alg{DT} & \alg{CRR} & \alg{OMIGA} & \alg{BC} & \alg{TD3} & \alg{OMAR} & \alg{STAR} & \alg{TD3+BC} & \alg{CRDT} & \alg{DT} & \alg{CRR} & \alg{OMIGA} & \alg{BC} & \alg{TD3} & \alg{OMAR} \\
\midrule
S1 & \metric{Ret.} & \best{648.3} & 570.9 & \second{576.8} & 512.3 & 516.5 & 384.9 & 395.8 & 410.0 & -675.0 & \best{1288.1} & \second{942.3} & 272.7 & 372.1 & 329.4 & 387.6 & 319.9 & 351.8 & -1569.2 \\
 & \metric{Comm.} & \best{3.63} & 3.27 & 3.37 & 3.29 & \second{3.49} & 3.25 & 3.29 & 3.24 & 3.37 & \best{6.76} & 5.74 & 5.65 & 5.55 & \second{5.82} & 5.74 & 5.63 & 5.64 & 5.37 \\
 & \metric{Pass} & \second{0.707} & \best{0.715} & \second{0.707} & 0.689 & 0.685 & 0.661 & 0.659 & 0.668 & 0.310 & \best{0.753} & \second{0.749} & 0.677 & 0.681 & 0.678 & 0.673 & 0.674 & 0.676 & 0.396 \\
 & \metric{Margin} & \best{0.0366} & \second{0.0341} & 0.0267 & 0.0237 & 0.0213 & 0.0194 & 0.0211 & 0.0181 & -0.0574 & \second{0.0572} & \best{0.0585} & 0.0311 & 0.0351 & 0.0341 & 0.0339 & 0.0316 & 0.0318 & -0.0401 \\
 & \metric{Coll. risk} & \best{0.028} & 0.063 & 0.073 & \second{0.058} & 0.080 & 0.126 & 0.118 & 0.116 & 0.106 & \best{0.333} & \second{0.424} & 0.766 & 0.655 & 0.742 & 0.632 & 0.700 & 0.677 & 0.878 \\
\midrule
S2 & \metric{Ret.} & \best{520.7} & -54.3 & \second{97.8} & 58.3 & -336.5 & -123.8 & -29.9 & -69.7 & -166.8 & \best{389.7} & -608.8 & \second{-592.4} & -622.6 & -1367.7 & -667.0 & -629.3 & -761.4 & -1341.3 \\
 & \metric{Comm.} & \second{3.33} & 2.74 & 3.05 & 2.92 & 2.50 & 2.82 & 2.84 & 2.81 & \best{3.61} & \second{4.39} & 3.63 & 4.07 & 3.91 & 3.57 & 3.91 & 3.92 & 3.81 & \best{4.72} \\
 & \metric{Pass} & \best{0.695} & 0.576 & \second{0.602} & 0.580 & 0.486 & 0.534 & 0.547 & 0.550 & 0.513 & \best{0.712} & 0.642 & \second{0.647} & 0.630 & 0.535 & 0.613 & 0.623 & 0.606 & 0.587 \\
 & \metric{Margin} & \best{0.0555} & 0.0066 & \second{0.0363} & 0.0356 & -0.0056 & 0.0225 & 0.0230 & 0.0240 & 0.0156 & \best{0.0814} & 0.0583 & \second{0.0654} & 0.0650 & 0.0348 & 0.0608 & 0.0645 & 0.0597 & 0.0481 \\
 & \metric{Coll. risk} & 0.127 & 0.133 & 0.159 & \second{0.112} & 0.119 & 0.154 & \best{0.092} & 0.139 & 0.261 & \best{0.732} & \second{0.896} & 1.175 & 1.070 & 1.159 & 1.024 & 1.075 & 1.080 & 1.694 \\
\midrule
S3 & \metric{Ret.} & \best{531.8} & \second{258.4} & 119.9 & -38.3 & 74.1 & -190.1 & -139.2 & -142.5 & -539.5 & \best{873.8} & \second{376.5} & -316.1 & -407.9 & -263.3 & -492.7 & -607.5 & -428.7 & -1557.7 \\
 & \metric{Comm.} & \best{3.47} & 2.90 & 2.77 & 2.41 & 2.57 & 2.45 & 2.44 & 2.55 & \second{3.14} & \best{6.10} & 5.07 & 4.89 & 4.64 & 4.77 & 4.47 & 4.47 & 4.55 & \second{5.35} \\
 & \metric{Pass} & \best{0.683} & \second{0.635} & 0.599 & 0.576 & 0.589 & 0.503 & 0.529 & 0.517 & 0.404 & \best{0.720} & \second{0.665} & 0.616 & 0.587 & 0.604 & 0.570 & 0.553 & 0.570 & 0.448 \\
 & \metric{Margin} & \best{0.0402} & \second{0.0075} & 0.0013 & -0.0056 & 0.0013 & -0.0130 & -0.0142 & -0.0130 & -0.0310 & \best{0.0689} & \second{0.0292} & 0.0161 & 0.0112 & 0.0145 & 0.0063 & 0.0006 & 0.0042 & -0.0154 \\
 & \metric{Coll. risk} & \best{0.018} & \second{0.057} & 0.089 & 0.116 & 0.062 & 0.068 & 0.092 & 0.079 & 0.179 & \second{0.409} & \best{0.325} & 0.868 & 0.717 & 0.677 & 0.657 & 0.691 & 0.613 & 1.105 \\
\bottomrule
\end{tabular}%
}
\end{table*}
\subsection{Random-Mobility Offline Policy Improvement}

We first test whether offline training can refine the initial behavior controller without reintroducing the online exploration risk in Section~II. Fig.~\ref{fig:train_support} shows that TD3 forms an initial controller, but its online curve remains noisy because value learning still depends on trial rollouts. After the dataset is frozen, imitation and conservative updates improve only part of the objective, whereas STAR-CRDT reaches the highest and most stable offline plateau. This supports the intended use of AERIS, where online interaction stops after an initial log exists and policy refinement continues offline.

\textbf{Ablation study.} We use CRDT as the main ablation of STAR-CRDT, as it removes support-aware candidate search and trust-gated distillation from the full design. The fourth panel of Fig.~\ref{fig:train_support} compares CRDT and STAR-CRDT using the same joint-action OOD ratio and ISAC violation cost as Fig.~\ref{fig:motivation}(b). CRDT remains closer to the log, but still leaves a higher violation cost because critic regularization alone does not identify which local correction is globally beneficial. Compared with CRDT and the baselines in Fig.~\ref{fig:motivation}(b), STAR-CRDT lowers violation cost while avoiding the large OOD shift caused by aggressive rectification, since it permits only a slightly larger but controlled support shift through critic-improving candidate selection and trusted-teacher distillation. Together with Fig.~\ref{fig:random_metrics}, where CRDT is stronger than plain DT but still trails STAR-CRDT, this ablation shows that centralized critic feedback must be coupled with support-aware rectification and trust-gated distillation to obtain reliable offline improvement.

The random-mobility results further verify that STAR-CRDT does more than avoid sensing-threshold and collision-risk violations. In Fig.~\ref{fig:random_metrics}, STAR-CRDT improves return by 29.3\% over TD3+BC, the strongest non-STAR baseline. It also improves the best baseline communication sum rate by 3.4\%, sensing pass rate by 4.8\%, and sensing margin by 69.1\%, while reducing collision-risk count by 54.2\%. These gains align with \textbf{P0} because the policy improves communication and sensing while reducing sensing-threshold violations and collision-risk events.

The remaining baselines confirm the imitation and OOD dilemma in Section~II-B. BC and DT remain close to the log and cannot reliably remove residual sensing violations. TD3+BC and CRR are safer than aggressive actor maximization, but their behavior regularization limits improvement once the logged controller is suboptimal. OMAR tends to rectify actions more aggressively and obtains poor return, which is consistent with the danger of leaving joint-action support. Thus, the comparison supports the offline-support improvement mechanism proved in Section~IV-C.

\subsection{System-Scale Evaluation}

We then test zero-shot system-scale transfer. The trained shared local actors are evaluated with $M=3$ to $6$ UAVs without additional offline retraining, while keeping two communication users and one sensing target per UAV. Thus, local observation/action semantics stay fixed but joint coordination becomes harder. STAR-CRDT ranks first at every scale. At $M=6$, it improves return by 39.3\% over TD3+BC and by 60.8\% over CRR, while maintaining the best sensing pass rate and sensing margin.

The scale trend shows that the local actor has learned a rule that remains useful as the joint-action space and pairwise safety constraints grow. Larger teams create more interference, coupled sensing beams, and collision-risk events. STAR-CRDT evaluates such team-level effects through the centralized critic before distillation, whereas conservative baselines inherit residual violations and aggressive rectification becomes less stable. This supports the CTDE choice in AERIS, where centralized information selects corrections during training and only the local sequence actor is deployed.

\subsection{Zero-Shot Road-Map Deployment}

\begin{figure}[!t]
\centering
\includegraphics[width=\linewidth]{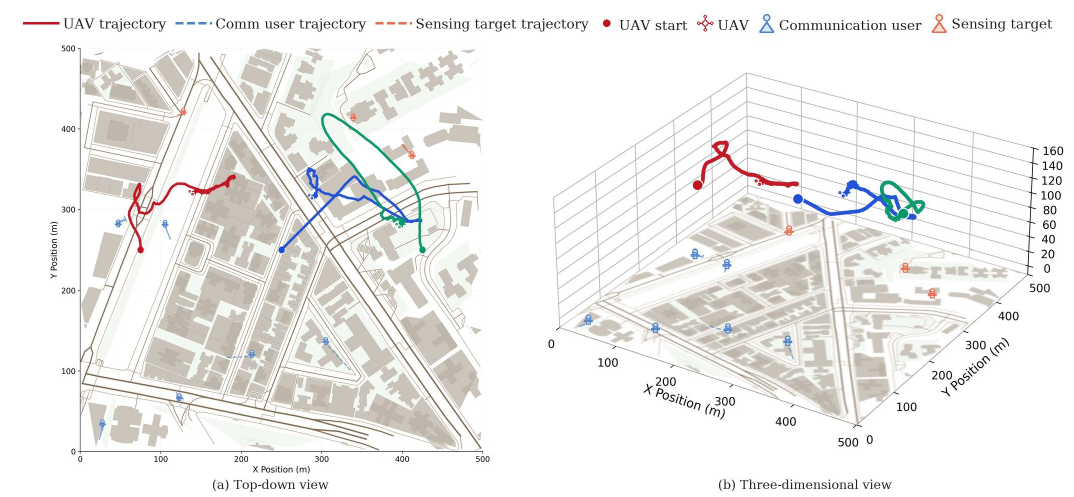}
\caption{Representative STAR-CRDT rollout under road-map mobility. The left and right panels show top-down and 3-D views, respectively.}
\label{fig:road_traj}
\end{figure}

Road-map deployment tests whether policies trained on the fixed random-mobility log can transfer to road-constrained urban mobility without adaptation. We deploy frozen policies on real OSM road graphs and Geofabrik OSM/GIS data extracts \cite{boeing2017osmnx,openstreetmap_copyright,geofabrik_osm_gis} using three $500\,\mathrm{m}\times500\,\mathrm{m}$ Hong Kong scenes: a dense grid (HK-OSM-S1, $72.6\,\mathrm{km}/\mathrm{km}^2$), an arterial corridor (HK-OSM-S2, $15.5\,\mathrm{km}/\mathrm{km}^2$), and a sparse local-road scene (HK-OSM-S3, $9.8\,\mathrm{km}/\mathrm{km}^2$). No road-scene rollout is used for training, early stopping, checkpoint selection, or hyperparameter tuning. We report both $M=3$ and $M=6$ to test increasing interference and safety coupling.

The representative rollout in Fig.~\ref{fig:road_traj} shows coordinated UAV motion around road-constrained users and targets. The 3-D view confirms feasible altitude control, and the trajectories balance communication coverage, sensing target tracking, and safety separation. Because road maps can concentrate users and targets along shared corridors, STAR-CRDT must learn a joint ISAC response rather than a purely communication-oriented motion.

Table~\ref{tab:road_eval} shows that STAR-CRDT obtains the best return in all six scene-scale settings. Across the 30 metric rows, it is best in 24 rows and remains top-two in 29 rows. At $M=6$, it improves return by 36.7\% in the dense grid and 132.1\% in the sparse-road scene over the strongest baseline. Its only non-top-two case is collision risk in arterial $M=3$, where BC has fewer risk events but much lower return, while OMAR sometimes gains communication rate at the cost of sensing and safety.

These results separate robustness from overfitting to random mobility. The three road structures create different route changes, demand concentrations, and coverage gaps, yet STAR-CRDT remains strongest because the centralized critic trains support-aware corrections instead of memorizing one mobility geometry. The transfer gains suggest that STAR-CRDT learns reusable local ISAC responses, including coverage recovery, sensing-margin preservation, and safety-aware separation. Overall, AERIS improves communication, sensing, and safety from fixed logs under decentralized execution.
\FloatBarrier

\section{Related Works}

\textbf{Multi-UAV ISAC and RL-based wireless control.} ISAC integrates communication and sensing, and UAVs add controllable aerial links with line-of-sight opportunities \cite{liu2022isac,zhang2021overview,liu2020joint_radar_comm,zeng2016uav,mozaffari2019tutorial}. Existing multi-UAV ISAC studies mainly optimize trajectories, beamforming, and sensing quality through model-based optimization or online learning \cite{lyu2023uav_isac,meng2022uav_periodic,gao2024marl_uav_isac,cheng2025radcom}. RL and MARL have also been used for adaptive wireless networking and resource allocation \cite{mao2016resource,xu2018experience,luong2019deep_rl_wireless,lowe2017maddpg,foerster2018coma,rashid2018qmix}. These works provide online or model-based baselines, while AERIS focuses on fixed-log improvement when new exploration is costly.

\textbf{Offline RL and offline MARL.} Offline RL learns from fixed datasets \cite{levine2020offline,fu2020d4rl,gulcehre2020rl_unplugged}. Main families include imitation such as BC, support-constrained value learning such as Batch-Constrained deep Q-learning (BCQ) and Bootstrapping Error Accumulation Reduction (BEAR) \cite{fujimoto2019bcq,kumar2019bear}, conservative or implicit value learning such as Conservative Q-Learning (CQL) and Implicit Q-Learning (IQL) \cite{kumar2020cql,kostrikov2022iql}, actor-regularized updates such as TD3+BC and CRR \cite{fujimoto2021td3bc,wang2020crr}, and sequence models such as DT \cite{chen2021decision_transformer}. Offline MARL extends fixed-dataset learning to coordinated agents with centralized training and decentralized policies. OMIGA uses implicit global-to-local value regularization, while OMAR searches critic-improving actor corrections \cite{wang2023omiga,pan2022omar}. Yet such methods may inherit imperfect logs or rely on unsupported critic gains. STAR-CRDT addresses both issues by searching for local corrections near logged actions and distilling them only when the centralized critic predicts reliable team-level improvement.

\section{Conclusion}

This paper presented AERIS, a novel offline policy improvement framework for distributed multi-UAV ISAC. AERIS reformulates trajectory-and-beamforming control as a fixed-log CTDE offline MARL problem, enabling local UAV execution while using logged global information to assess team-level communication, sensing, and safety effects. To solve this problem, we developed STAR-CRDT, a new algorithm that distills support-aware teacher corrections selected by centralized critic feedback. This lets AERIS improve beyond imperfect logs without degenerating into pure imitation or unsupported critic maximization. With an offline-support improvement guarantee and evaluations under random mobility, system-scale transfer, and unseen OSM road maps, AERIS consistently improves return, communication, sensing, and safety from fixed flight logs before deployment.

\bibliographystyle{IEEEtran}
\bibliography{references}

\end{document}